\ifx\SoCGVer\undefined
   \newcommand{\InSoCGVer}[1]{}%
   \newcommand{\InNotSoCGVer}[1]{#1}%
\else
   \def\ProofInAppendix{TRUE}
   \newcommand{\InSoCGVer}[1]{#1}%
   \newcommand{\InNotSoCGVer}[1]{}%
\fi

\InSoCGVer{%
   \documentclass[a4paper,UKenglish]{lipics}
}
\InNotSoCGVer{%
   \documentclass[12pt]{article}%
}

\InNotSoCGVer{%
   \usepackage[cm]{fullpage}%
}

\usepackage{mleftright}%
\usepackage{proofapnd}%
\usepackage{graphicx}%
\usepackage{amsmath}%
\usepackage{xspace}

\usepackage{xcolor}

\usepackage{picins}

\usepackage{amsmath}%
\InNotSoCGVer{%
   \usepackage[amsmath,thmmarks]{ntheorem}
   \theoremseparator{.}%
}%

\usepackage{amssymb}%
\InNotSoCGVer{%
   \usepackage{titlesec}%
   \titlelabel{\thetitle. }%
}

\usepackage{paralist}%

\InNotSoCGVer{%
   \usepackage{euscript}%
}

\usepackage{caption}%

\usepackage{hyperref}%
\hypersetup{%
   breaklinks,%
   colorlinks=true,%
   linkcolor=[rgb]{0.45,0.0,0.0},%
   citecolor=[rgb]{0,0,0.45}
}

\numberwithin{figure}{section}%
\numberwithin{table}{section}%
\numberwithin{equation}{section}%

\InNotSoCGVer{%
   \theoremstyle{theorem} \newtheorem{theorem}{Theorem}[section]
   \newtheorem{lemma}[theorem]{Lemma}%
   \newtheorem{corollary}[theorem]{Corollary}
   \newtheorem{observation}[theorem]{Observation}
   
   \theoremstyle{plain}%
   \theoremheaderfont{\sf} \theorembodyfont{\upshape}%
   \theoremseparator{.}%
   
   \newtheorem{defn}[theorem]{Definition}
   
   \newtheorem*{remark:unnumbered}[theorem]{Remark}%
   \theoremheaderfont{\em}%
   \theorembodyfont{\upshape}%
   \theoremstyle{nonumberplain} \theoremseparator{}
   \theoremsymbol{\myqedsymbol} \newtheorem{proof}{Proof:}
   
}

\InNotSoCGVer{
   \newcommand{\myparagraph}[1]{\paragraph{#1}}
}
\InSoCGVer{
   \newcommand{\myparagraph}[1]{\bigskip\noindent{\textbf{#1}}}
}

\InSoCGVer{%
   \newtheorem{observation}[theorem]{Observation} %

   \theoremstyle{definition}%
   \newtheorem{defn}[theorem]{Definition}%

   \theoremstyle{remark}%
   \newtheorem{remark*}[theorem]{Remark}%
}

\newcommand{\atgen}{\symbol{'100}}
\newcommand{\SarielThanks}[1]{\thanks{Department of Computer
      Science; 
      University of Illinois; 
      201 N. Goodwin Avenue;
      Urbana, IL, 61801, USA;
      {\tt sariel\atgen{}illinois.edu}; {\tt
         \url{http://sarielhp.org/}.} #1}}

\newcommand{\NirmanThanks}[1]{%
   \thanks{%
      Department of Computer Science; %
      University of California; %
      2120B Harold Frank Hall; %
      Santa Barbara, CA, 93106, USA; %
      {\tt nirman\atgen{}cs.ucsb.edu}; #1}}

\newcommand{\ceil}[1]{\left\lceil {#1} \right\rceil}

\definecolor{blue25}{rgb}{0, 0, 11}
\newcommand{\emphic}[2]{%
   \textcolor{blue25}{%
      \textbf{\emph{#1}}}%
   \index{#2}}

\newcommand{\emphi}[1]{\emphic{#1}{#1}}

\newcommand{\diameterX}[1]{\mathrm{d{}i{}am}\pth{#1}}

\newcommand{\permut}[1]{\left\langle {#1} \right\rangle}
\newcommand{\PrmX}[1]{\permut{#1}}

\newcommand{\norm}[1]{\left\| {#1}  \right\|}
\newcommand{\distY}[2]{\norm{#1 - #2}}
\newcommand{\brc}[1]{\left\{ {#1} \right\}}

\renewcommand{\th}{th\xspace}

\DefineNamedColor{named}{RedViolet} {cmyk}{0.07,0.90,0,0.34}

\providecommand{\AlgorithmI}[1]{{%
      \textcolor[named]{RedViolet}{\texttt{\bf{#1}}}}}
\providecommand{\Algorithm}[1]{{%
      \AlgorithmI{#1}\index{algorithm!#1@{\AlgorithmI{#1}}}}}

\newcommand{\cPnt}{\overline{\mathrm{c}}}

\DefineNamedColor{named}{OliveGreen}    {cmyk}{0.64,0,0.95,0.40}

\providecommand{\ComplexityClass}[1]{{{\textcolor[named]{OliveGreen}{%
      \textsc{#1}}}}}
\providecommand{\NPHard}{{\ComplexityClass{NP-Hard}}\xspace}

\newcommand{\pth}[1]{\mleft({#1}\mright)}

\newcommand{\diam}{\Delta}
\newcommand{\aDiam}{\Delta'}

\newcommand{\CHChar}{{\mathcal{CH}}}
\newcommand{\CHX}[1]{\CHChar\pth{#1}}
\newcommand{\CH}{\mathcal{C}}

\newcommand{\cone}{\phi}

\newcommand{\eps}{{\varepsilon}}%
\newcommand{\PntSet}{\mathsf{P}}
\newcommand{\CenSet}{\mathsf{C}}

\newcommand{\OptSet}{\mathcal{O}}
\newcommand{\ocen}{o}

\newcommand{\SuffixSet}{\mathsf{S}}

\newcommand{\query}{\mathsf{q}}

\newcommand{\ball}{b}%
\newcommand{\ballY}[2]{\mathrm{ball}\pth{#1, #2}}%
\newcommand{\nnY}[2]{\mathrm{nn}\pth{#1, #2}}%
\newcommand{\GSet}{G}%

\newcommand{\ropt}{r_{\mathrm{opt}}}%
\newcommand{\roptX}[1]{\ropt^{#1}}%

\newcommand{\BNN}{T_{nn}}
\newcommand{\BANN}{T_{ann}}

\newcommand{\nnp}{\nu}%
\newcommand{\pnt}{\mathsf{p}}%
\newcommand{\pntA}{\mathsf{x}}%
 \newcommand{\pntB}{\mathsf{z}}
\newcommand{\pntC}{\mathsf{y}}
\newcommand{\pntD}{\mathsf{u}}

\newcommand{\cen}{\mathsf{c}}

\newcommand{\sep}[1]{\,\left|\, {#1} \right.}
\newcommand{\distSet}[2]{d\pth{#1, #2}}

\newcommand{\PartitionY}[2]{\mathrm{\Pi}\pth{#1, #2}}

\newcommand{\cardin}[1]{\left| {#1} \right|}%

\newcommand{\Domain}{\mathcal{D}}

\newcommand{\myqedsymbol}{\rule{2mm}{2mm}}

\newcommand{\DD}{\mathsf{d}}

\renewcommand{\Re}{{\mathbb{R}}}%

\newcommand{\Term}[1]{\textsf{#1}}%
\newcommand{\ANN}{\Term{ANN}\xspace}%
\newcommand{\NN}{\Term{NN}\xspace}%
\newcommand{\VC}{\Term{VC}\xspace}%

\newcommand{\Set}[2]{\left\{ #1 \;\middle\vert\; #2 \right\}}

\newcommand{\DotProd}[2]{\permut{{#1},{#2}}}

\newcommand{\RSample}{\mathsf{R}}

\newcommand{\iterX}[1]{\mathrm{iter}\pth{#1}}

\newcommand{\KCenPrc}[2]{\mathrm{price}\pth{{#1}, {#2} }}

\newcommand{\Copt}{C_\mathrm{opt}}

\newcommand{\coneC}{\cone^{}_C}%
\newcommand{\ConeSet}{\EuScript{C}}

\newcommand{\bd}{\partial}
\newcommand{\AlgGPermutNN}{\Algorithm{GreedyPermutNN}\xspace}

\newcommand{\Spread}{\Phi}

\newcommand{\ray}{\psi}
\newcommand{\vVec}{\mathsf{v}}

\newcommand{\cconst}{c}

\newcommand{\si}[1]{#1}%
\newcommand{\RetInHull}{\textsc{In}\xspace}
\newcommand{\RetOutHull}{\textsc{Out}\xspace}

\providecommand{\GreedyKCenter}{\Algorithm{GreedyKCenter}\xspace}
\newcommand{\ds}{\displaystyle}
\newcommand{\surface}{S}%
\newcommand{\object}{\mathcal{O}}%
\newcommand{\asurface}{\widehat{S}}%
\newcommand{\AFM}{\textrm{AFM}\xspace}%

\newcommand{\etal}{\textit{et~al.}\xspace}

\newcommand{\XSays}[2]{%
   \textcolor{brown}{%
      $\rule[-0.12cm]{0.2in}{0.5cm}$%
      \fbox{\small \tt #1:} }%
   {\textcolor{blue}{{\sc \small{#2}}}}
   \marginpar{\textcolor{red}{#1}}
   \textcolor{brown}{%
      {$\rule[0.1cm]{0.3in}{0.1cm}$\fbox{\tt
            end}$\rule[0.1cm]{0.3in}{0.1cm}$}%
   }%
}

\ifx\HideComments\undefined%
\else%
\renewcommand{\XSays}[2]{}%
\fi%

\newcommand{\HLinkShort}[2]{\hyperref[#2]{#1\ref*{#2}}}
\newcommand{\HLink}[2]{\hyperref[#2]{#1~\ref*{#2}}}
\newcommand{\HLinkPage}[2]{\hyperref[#2]{#1~\ref*{#2}%
      $_\text{p\pageref{#2}}$}}
\newcommand{\HLinkPageOnly}[1]{\hyperref[#1]{Page~\refpage*{#1}%
      $_\text{p\pageref{#1}}$}}

\newcommand{\HLinkSuffix}[3]{\hyperref[#2]{#1\ref*{#2}{#3}}}
\newcommand{\HLinkPageSuffix}[3]{\hyperref[#2]{#1\ref*{#2}%
      #3$_\text{p\pageref{#2}}$}}

\newcommand{\figlab}[1]{\label{fig:#1}}
\newcommand{\figref}[1]{\HLink{Figure}{fig:#1}}

\newcommand{\seclab}[1]{\label{sec:#1}}
\newcommand{\secref}[1]{\HLink{Section}{sec:#1}}

\newcommand{\lemlab}[1]{\label{lemma:#1}}
\newcommand{\lemref}[1]{\HLink{Lemma}{lemma:#1}}%

\newcommand{\thmlab}[1]{{\label{theo:#1}}}
\newcommand{\thmref}[1]{\HLink{Theorem}{theo:#1}}

\providecommand{\deflab}[1]{\label{def:#1}}

\newcommand{\itemlab}[1]{\label{item:#1}}
\newcommand{\itemref}[1]{\HLinkSuffix{(}{item:#1}{)}}

\providecommand{\eqlab}[1]{}%
\renewcommand{\eqlab}[1]{\label{equation:#1}}
\newcommand{\Eqref}[1]{\HLinkSuffix{Eq.~(}{equation:#1}{)}}

\begin{document}

\title{Space Exploration via Proximity Search%
   \InSoCGVer{%
      \footnote{%
         Work on this paper by S.H. and B.R. was partially supported
         by NSF AF awards CCF-1421231, and % Started June 2014
         CCF-1217462.  % Started June 2012
         N.K.  was partially supported by a NSF AF award CCF-1217462
         while at UIUC, and by NSF grant CCF-1161495 and a grant from
         \si{DARPA} while at \si{UCSB}.%
         D.M.  was partially supported by NSF award CCF-1117259 and
         \si{ONR} award N00014-08-1-1015.%
      }%
   }%
}

\InSoCGVer{%
   \author[1]{Sariel Har-Peled}%
   \author[2]{Nirman Kumar}%
   \author[3]{David M. Mount}%
   \author[1]{Benjamin Raichel}%
   \affil[1]{%
      Department of Computer Science, %
      University of Illinois\\ %
      201 N. Goodwin Avenue, %
      Urbana, IL, 61801, USA.\\ %
      \texttt{{sariel}@\si{illinois}.\si{edu}}, %
      \texttt{\si{raichel}2@\si{illinois}.\si{edu}.} %
   }%
   \affil[2]{%
      Department of Computer Science, %
      University of California \\%
      2120B Harold Frank Hall, %
      Santa Barbara, CA, 93106, USA. \\ %
      {\tt \si{nirman\atgen{}cs.\si{ucsb.edu}}}.%
   }%
   \affil[3]{%
      Department of Computer Science, %
      University of Maryland \\ %
      College Park, MD, 20742, USA. \\%
      \si{\tt mount\atgen{}cs.\si{umd}.\si{edu}}.  }%

}

\InNotSoCGVer{%
   \author{%
      Sariel Har-Peled\SarielThanks{%
         Work on this paper was partially supported by a NSF AF award
         CCF-0915984.}%
      \and%
      Nirman Kumar%
      \NirmanThanks{%
         Work on this paper was partially supported by a NSF AF award
         CCF-1217462 while the author was a student at UIUC, and by
         NSF grant CCF-1161495 and a grant from \si{DARPA} while the
         author has been a postdoc at \si{UCSB}.%
      }%
      \and
      David M. Mount%
      \thanks{Department of Computer Science; %
         University of Maryland; %
         College Park, MD, 20742, USA; %
         \si{\tt mount\atgen{}cs.\si{umd}.\si{edu}}; %
         {\tt \url{http://www.cs.umd.edu/\string~mount/}.}  %
         Work on this paper was partially supported by NSF award
         CCF-1117259 and \si{ONR} award N00014-08-1-1015.}
      \and%
      Benjamin Raichel%
   }%
}

\InNotSoCGVer{\date{\today}}%
\InSoCGVer{%
   \authorrunning{S. Har-Peled, N. Kumar, D. M. Mount and B. Raichel}
   
   \Copyright{Sariel Har-Peled, Nirman Kumar, David Mount %
      and Benjamin Raichel}%
   
   % 
   % mandatory: Please choose ACM 1998 classifications from
   % http://www.acm.org/about/class/ccs98-html . E.g., cite
   % as "F.1.1 Models of Computation".

   \serieslogo{}%please provide file name (without suffix)
   \volumeinfo%(easy chair interface)
   {}% editors
   {2}% number of editors: 1, 2, ....
   {Conference title on which this volume is based on}% event
   {1}% volume
   {1}% issue
   {1}% starting page number
   \EventShortName{}
   \DOI{10.4230/LI{}PI{}cs.xxx.yyy.p}% to be completed by the volume editor
}

\maketitle

\begin{abstract}
    We investigate what computational tasks can be performed on a
    point set in $\Re^d$, if we are only given black-box access to it
    via nearest-neighbor search. This is a reasonable assumption if
    the underlying point set is either provided implicitly, or it is
    stored in a data structure that can answer such queries. In
    particular, we show the following: \smallskip
    \begin{compactenum}[\quad(A)]
        \item One can compute an approximate bi-criteria $k$-center
        clustering of the point set, and more generally compute a
        greedy permutation of the point set.
        
        \item One can decide if a query point is (approximately)
        inside the convex-hull of the point set.
    \end{compactenum}
    \smallskip %
    We also investigate the problem of clustering the given point
    set, such that meaningful proximity queries can be carried out on
    the centers of the clusters, instead of the whole point set.

\end{abstract}

%%%%%%%%%%%%%%%%%%%%%%%%%%%%%%%%%%%%%%%%%%%%%%%%%%%%%%%%%%%%%%%%%% 
%%%%%%%%%%%%%%%%%%%%%%%%%%%%%%%%%%%%%%%%%%%%%%%%%%%%%%%%%%%%%%%%%% 

\section{Introduction}

Many problems in Computational Geometry involve sets of points in
$\Re^d$. Traditionally, such a point set is presented explicitly, say,
as a list of coordinate vectors. There are, however, numerous
applications in science and engineering where point sets are presented
\emph{implicitly}. This may arise for various reasons: (1) the point
set (which might be infinite) is a physical structure that is
represented in terms of a finite set of sensed measurements such as a
point cloud, (2) the set is too large to be stored explicitly in
memory, or (3) the set is procedurally generated from a highly
compressed form. (A number of concrete examples are described below.)

Access to such an implicitly-represented point set $\PntSet$ is
performed through an \emph{oracle} that is capable of answering
queries of a particular type. We can think of this oracle as a
black-box data structure, which is provided to us in lieu of an
explicit representation. Various types of probes have been studied
(such as finger probes, line probes, and X-ray probes
\cite{s-pgp-89}). Most of these assume that $\PntSet$ is connected
(e.g., a convex polygon) and cannot be applied when dealing with
arbitrary point sets. In this paper we consider a natural choice for
probing general point sets based on computing nearest neighbors, which
we call \emph{proximity probes}.

More formally, we assume that the point set $\PntSet$ is a (not
necessarily finite) compact subset of $\Re^d$. The point set $\PntSet$
is accessible only through a nearest-neighbor data structure, which
given a query point $\query$, returns the closest point of $\PntSet$
to $\query$. Some of our results assume that the data structure
returns an exact nearest neighbor (\NN) and others assume that the
data structure returns a $(1+\eps)$-approximate nearest-neighbor
(\ANN). (See \secref{prelims} for definitions.) In any probing
scenario, it is necessary to begin with a general notion of the set's
spatial location. We assume that $\PntSet$ is contained within a given
compact subset $\Domain$ of $\Re^d$, called the \emph{domain}.

The oracle is given as a black-box. Specifically, we do not allow
deletions from or insertions into the data structure. We do not assume
knowledge of the number of data points, nor do we make any continuity
or smoothness assumptions. Indeed, most of our results apply to
infinite point sets, including volumes or surfaces.

\subsection*{Prior Work and Applications}
Implicitly-represented point sets arise in various applications. One
example is that of analyzing a geometric shape through probing. An
example of this is Atomic Force Microscopy (\AFM) \cite{w-afm-14}.
This technology can reveal the undulations of a surface at the
resolution of fractions of a nanometer. It relies on the principle
that when an appropriately designed tip (the probe) is brought in the
proximity of a surface to scan it, certain atomic forces minutely
deflect the tip in the direction of the surface. Since the deflection
of the tip is generally to the closest point on the surface, this mode
of acquisition is an example of proximity probing. A sufficient number
of such samples can be used to reconstruct the surface
\cite{bqg-afm-86}.

The topic of shape analysis through probing has been well studied
within the field of computational geometry. The most commonly assumed
probe is a \emph{finger probe}, which determines the first point of
contact of a ray and the set. Cole and Yap \cite{cy-sfp-87} pioneered
this area by analyzing the minimum number of finger probes needed to
reconstruct a convex polygon. Since then, various alternative probing
methods have been considered. For good surveys of this area, see
Skiena \cite{s-pgp-89,s-grp-97}.

More recently, Boissonnat \etal~\cite{bgo-lssp-07} presented an
algorithm for learning a smooth unknown surface $\surface$ bounding an
object $\object$ in $\Re^3$ through the use of finger probes. Under
some reasonable assumptions, their algorithm computes a triangulated
surface $\asurface$ that approximates $\surface$ to a given level of
accuracy. In contrast to our work, which applies to general point
sets, all of these earlier results assume that the set in question is
a connected shape or surface.

Implicitly-represented point sets also arise in geometric modeling.
Complex geometric sets are often generated from much smaller
representations. One example are fractals sets, which are often used
to model natural phenomena such as plants, clouds, and terrains
\cite{skgtb-spmtm-09}. Fractal objects are generated as the limit of
an iterative process \cite{m-fgn-83}. Due to their regular, recursive
structure it is often possible to answer proximity queries about such
a set without generating the set itself.

Two other examples of infinite sets generated implicitly from finite
models include (1) subdivision surfaces \cite{at-imss-10}, where a
smooth surface is generated by applying a recursive refinement process
to a finite set of boundary points, and (2) metaballs
\cite{b-gasd-82}, where a surface is defined by a blending function
applied to a collection of geometric balls. In both cases, it is
possible to answer nearest neighbor queries for the underlying object
without the need to generate its boundary.

Proximity queries have been applied before. Panahi
\etal~\cite{pasfg-eppam-13} use proximity probes on a convex polygon
in the plane to reconstruct it exactly. Goel
\etal~\cite{giv-rahdp-01}, reduce the approximation versions of
several problems like diameter, farthest neighbors, discrete center,
metric facility location, bottleneck matching and minimum weight
matching to nearest neighbor queries. They sometimes require other
primitives for their algorithms, for example computation of the
minimum enclosing ball or a dynamic version of the approximate
nearest-neighbor oracle. Similarly, the computation of the minimum
spanning tree \cite{him-anntr-12} can be done using nearest-neighbor
queries (but the data structure needs to support deletions). For more
details, see the survey by Indyk \cite{i-nnhds-04}.

\subsection*{Our contributions}
In this paper we consider a number of problems on
implicitly-represented point sets. Here is a summary of our main
results.

\myparagraph{$k$-center clustering and the greedy permutation.} %
Given a point set $\PntSet$, a \emph{greedy permutation} (informally)
is an ordering of the points of $\PntSet$: $\pnt_1, \ldots, \pnt_k,
\ldots,$ such that for any $k$, the set of points $\brc{\pnt_1,
\ldots, \pnt_k }$ is a $O(1)$-approximation to the optimal $k$-center
clustering. This sequence arises in the $k$-center approximation of
Gonzalez \cite{g-cmmid-85}, and its properties were analyzed by
Har-Peled and Mendel \cite{hm-fcnld-06}. Specifically, if $\PntSet$
can be covered by $k$ balls of radius $r_k$, then the maximum distance
of any point of $\PntSet$ to its nearest neighbor in $\brc{\pnt_1,
\ldots, \pnt_k}$ is $O(r_k)$. 

In \secref{k:cent:clust}, we show that under reasonable assumptions,
in constant dimension, one can compute such an approximate greedy
permutation using $O(k)$ exact proximity queries. If the oracle
answers $(1+\eps)$-\ANN queries, then for any $k$, the permutation
generated is competitive with the optimal $k$-center clustering,
considering the first $O\pth{ k \log_{1/\eps} \Spread}$ points in this
permutation, where $\Spread$ is (roughly) the spread of the point set.
The hidden constant factors grow exponentially in the dimension.

\myparagraph{Approximate convex-hull membership.} %
Given a point set $\PntSet$ in $\Re^d$, we consider the problem of
deciding whether a given query point $\query \in \Re^d$ is inside
its convex-hull $\CH=\CHX{\PntSet}$. We say that the answer is
$\eps$-approximately correct if the answer is correct whenever the
query point's distance from the boundary of $\CH$ is at least $\eps
\cdot \diameterX{\CH}$. In \secref{ch:memb}, we show that, given an
oracle for $(1+\eps^2/c)$-\ANN queries, for some sufficiently large
constant $c$, it is possible to answer approximate convex-hull
membership queries using $O(1/\eps^2)$ proximity queries. Remarkably,
the number of queries is independent of the dimension of the data.

Our algorithm operates iteratively, by employing a gradient
descent-like approach. It generates a sequence of points, all within
the convex hull, that converges to the query point. Similar techniques
have been used before, and are sometimes referred to as the
Frank-Wolfe algorithm. Clarkson provides a survey and some new results
of this type \cite{c-csgaf-10}. A recent algorithm of this type is the
work by Kalantari \cite{k-ctach-12}.  Our main new contribution for
the convex-hull membership problem is showing that the iterative
algorithm can be applied to implicit point sets using nearest-neighbor
queries.

\myparagraph{Balanced proximity clustering.} %
We study a problem that involves summarizing a point set in a way that
preserves proximity information. Specifically, given a set $\PntSet$
of $n$ points in $\Re^d$, and a parameter $k$, the objective is to
select $m$ centers from $\PntSet$, such that if we assign every point
of $\PntSet$ to its nearest center, no center has been selected by
more than $k$ points. This problem is related to topic of capacitated
clustering from operations research \cite{mb-sccp-84}. 

In \secref{density:clust}, we show that in the plane there exists such
a clustering consisting of $O(n/k)$ such centers, and that in higher
dimensions one can select $O((n/k) \log (n/k))$ centers (where the
constant depends on the dimension). This result is not directly
related to the other results in the paper.
 
\myparagraph{Paper organization.} %
In \secref{prelims} we review some relevant work on $k$-center
clustering. In \secref{k:cent:clust} we provide our algorithm to
compute an approximate $k$-center clustering. In \secref{ch:memb} we
show how we can decide approximately if a query point is within the
convex hull of the given data points in a constant number of queries,
where the constant depends on the degree of accuracy desired. Finally,
in \secref{density:clust} we investigate balanced Voronoi partitions,
which provides a density-based clustering of the data. Here we assume
that all the data is known and the goal is to come up with a useful
clustering that can help in proximity search queries.

\section{Preliminaries}
\seclab{prelims}

\subsection{Background --- $k$-center clustering and %
   the greedy permutation}

The following is taken from \cite[Chap. 4]{h-gaa-11}, and is provided
here for the sake of completeness.

In the $k$-center clustering problem, a set $\PntSet \subseteq \Re^d$
of $n$ points is provided together with a parameter $k$. The objective 
is to find a set of $k$ points, $\CenSet \subseteq \PntSet$, such that
the maximum distance of a point in $\PntSet$ to its closest point in
$\CenSet$ is minimized.  Formally, define
\begin{math}
    \KCenPrc{\CenSet}{\PntSet}%
    =%
    \max_{\pnt \in \PntSet}\min_{\cen \in \CenSet} \distY{\pnt}{\cen}.
\end{math}
Let $\Copt$ denote the set of centers achieving this minimum. The
$k$-center problem can be interpreted as the problem of computing the
minimum radius, called the \emph{$k$-center clustering radius}, such
that it is possible to cover the points of $\PntSet$ using $k$ balls
of this radius, each centered at one of the data points. It is known
that $k$-center clustering is \NPHard. Even in the plane, it is
\NPHard to approximate to within a factor of $\pth{1+\sqrt{7}}/2
\approx 1.82$ \cite{fg-oafac-88}.

\myparagraph{The greedy clustering algorithm.} %
Gonzalez \cite{g-cmmid-85} provided a $2$-approximation algorithm for
$k$-center clustering.  This algorithm, denoted by \GreedyKCenter,
repeatedly picks the point farthest away from the current set of
centers and adds it to this set. Specifically, it starts by picking an
arbitrary point, $\cPnt_1$, and setting $\CenSet_1 = \brc{
   \cPnt_1}$. For $i > 1$, in the $i$\th iteration, the algorithm
computes
\begin{align}
    r_{i-1}%
    =%
    \KCenPrc{\CenSet_{i-1}}{\PntSet}%
    =%
    % \max_{\pnt \in \PntSet} d_{i-1} [ \pnt] =
    \max_{\pnt \in \PntSet} \distSet{\pnt}{\CenSet_{i-1}}
    \eqlab{radius:k:center}%
\end{align}
and the point $\cPnt_{i}$ that realizes it, where
\begin{math}
    \distSet{\pnt}{\CenSet_{i-1}} = \min_{\cen \in \CenSet_{i-1}}
    \distY{\pnt}{\cen}.
\end{math}
 Next, the algorithm adds $\cPnt_i$ to $\CenSet_{i-1}$ to form the new
set $\CenSet_{i}$. This process is repeated until $k$ points have been
collected.

If we run \GreedyKCenter till it exhausts all the points of $\PntSet$
(i.e., $k=n$), then this algorithm generates a permutation of
$\PntSet$; that is,
$\PrmX{\PntSet} = \permut{\cPnt_1, \ldots, \cPnt_n}$. We will
refer to $\PrmX{\PntSet}$ as the \emphi{greedy permutation} of
$\PntSet$.  There is also an associated sequence of radii
$\permut{r_1, \ldots, r_n}$, and the key property of the greedy
permutation is that for each $i$ with $1 \leq i \leq n$, all the
points of $\PntSet$ are within a distance at most $r_i$ from the
points of $\CenSet_i = \permut{\cPnt_1, \ldots, \cPnt_i}$.  The greedy
permutation has applications to packings, which we describe next.

\begin{defn}
    \deflab{packing}%
    A set $S \subseteq \PntSet$ is an \emphic{$r$-packing}{packing}
    for $\PntSet$ if the following two properties hold:
    \begin{compactenum}[(i)]
        \item \emphic{Covering property}{covering property}: All the
        points of $\PntSet$ are within a distance at most $r$ from the
        points of $S$.
        
        \item \emphic{Separation property}{separation property}: For
        any pair of points $\pnt, \pntA \in S$, $\distY{\pnt}{ \pntA}
        \geq r$.
    \end{compactenum}
    (For most purposes, one can relax the separation property by
    requiring that the points of $S$ be at distance $\Omega(r)$
    from each other.)
\end{defn}

Intuitively, an $r$-packing of a point set $\PntSet$ is a compact
representation of $\PntSet$ at resolution $r$. Surprisingly, the
greedy permutation of $\PntSet$ provides us with such a representation
for all resolutions.

\begin{lemma}[\cite{h-gaa-11}]%
    \lemlab{greedy:clustering}%
    \begin{inparaenum}[(A)]
        \item Let $\PntSet$ be a set of $n$ points in $\Re^d$, and let
        its greedy permutation be $\langle \cPnt_1, \ldots,$
        $ \cPnt_n \rangle$ with the associated sequence of radii
        $\permut{r_1, \ldots, r_n}$. For any $i$, 
        $\CenSet_i = \permut{\cPnt_1, \ldots, \cPnt_i}$ is an
        $r_i$-packing of $\PntSet$. Furthermore, $r_i$ is a 
        $2$-approximation for the optimal $i$-center clustering radius 
        of $\PntSet$.
        
        \item \itemlab{packing}%
        For any $k$, let $\roptX{k}$ be the radius of the optimal
        $k$-center clustering of $\PntSet$. Then, for any constant
        $c$,
        \begin{math}
            \roptX{O(c^d k)} \leq \roptX{k}/c.
        \end{math}
        
        \item Computing the optimal $k$-center clustering of the first
        $O(k /\eps^d)$ points of the greedy permutation, after
        appropriate rescaling, results in a $(1+\eps)$-approximation to
        the optimal $k$-center clustering of $\PntSet$.
        
    \end{inparaenum}
\end{lemma}

\subsection{Setup}
Our algorithms operate on a (not necessarily finite) point set $\PntSet$ in
$\Re^d$. We assume that we are given a compact subset of $\Re^d$,
called the \emphi{domain} and denoted $\Domain$, such that $\PntSet
\subseteq \Domain$. Throughout we assume that $\Domain$ is the unit
hypercube $[0,1]^d$.  The set $\PntSet$ (not necessarily finite) is
contained in $\Domain$. 

Given a query point $\query \in [0,1]^d$, let $\nnY{\query}{\PntSet} =
\arg \min_{\pnt\in\PntSet} \distY{\query}{\pnt}$ denote the nearest
neighbor (\NN) of $\query$. We say a point $\pntA$ is a
$(1+\eps)$-approximate nearest-neighbor (\ANN) for $\query$ if
$\distY{\query}{\pntA} \leq (1+\eps)
\distY{\query}{\nnY{\query}{\PntSet}}$. We assume that the sole access
to $\PntSet$ is through ``black-box'' data structures $\BNN$ and
$\BANN$, which given a query point $\query$, return the \NN and \ANN,
respectively, to $\query$ in $\PntSet$.

\section{Using proximity search to compute %
   $k$-center clustering}
\seclab{k:cent:clust}

\myparagraph{The problem.} %
Our purpose is to compute (or approximately compute) a $k$-center
clustering of $\PntSet$ through the \ANN black box we have, where $k$
is a given parameter between $1$ and $n$.

\subsection{Greedy permutation via \NN queries: %
   \AlgGPermutNN}
\seclab{g:p:alg}

Let $\query_0$ be an arbitrary point in $\Domain$. Let $\nnp_0$ be its
nearest-neighbor in $\PntSet$ computed using the provided \NN
data structure $\BNN$. Let
$\ball_0 = \ballY{\query_0}{ \distY{\query_0}{\nnp_0} }$ be the open
ball of radius $\distY{\query_0}{\nnp_0}$ centered at
$\query_0$. Finally, let $\GSet_0 = \brc{\nnp_0}$, and let
$\Domain_0 = \Domain \setminus \ball_0$.

In the $i$\th iteration, for $i > 0$, let $\query_i$ be the point in
$\Domain_{i-1}$ farthest away from $\GSet_{i-1}$. Formally, this is
the point in $\Domain_{i-1}$ that maximizes
$\distSet{\query_i}{\GSet_{i-1}}$, where $\distSet{\query}{X} =
\min_{\cen \in X} \distY{\cen}{\query}$. Let $\nnp_i =
\nnY{\query_i}{\PntSet}$ denote the nearest-neighbor $\nnp_i$ to
$\query_i$ in $\PntSet$, computed using $\BNN$. Let
\begin{align*}
%    r_i = \distY{\query_i}{\nnp_i},%
    r_i = \distSet{\query_i}{\GSet_{i-1}}, %
    \quad%
    \ball_i = \ballY{\query_i}{r_i},%
    \quad%
    \GSet_i = \GSet_{i-1} \cup \brc{\nnp_i},%
    \quad \text{and} \quad \Domain_i = \Domain_{i-1} \setminus
    \ball_i.
\end{align*}

Left to its own devices, this algorithm computes a sequence of not
necessarily distinct points $\nnp_0, \nnp_1, \ldots$ of $\PntSet$.  If
$\PntSet$ is not finite then this sequence may also have infinitely
many distinct points.  Furthermore, $\Domain_0 \supseteq \Domain_1
\supseteq \ldots$ is a sequence of outer approximations to $\PntSet$.
% that converges to $\PntSet$.

The execution of this algorithm is illustrated in
\figref{g:p:execution}.

\newcommand{\PicStep}[1]{%
   \begin{minipage}{0.22\linewidth}
       \hspace{-0.35cm}%
       {\includegraphics[page=#1,%
          trim = 0.01mm 0.01mm 0.01mm 0.01mm, clip, %
          width=1.08\linewidth]{figs/domain}%
       } \hspace{-1cm}
   \end{minipage}
}
\begin{figure}[p]
    \centerline{%
       \begin{minipage}{0.98\linewidth}
           \begin{minipage}{1.01\linewidth}
               \begin{tabular}{cc%
                 cc}
                   \PicStep{1}&%
                   \PicStep{2}&%
                   \PicStep{3}&%
                   \PicStep{4}\\%
                   (1) & (2) & (3) & (4) \\
                   \PicStep{5}&%
                   \PicStep{6}&%
                   \PicStep{7}&%
                   \PicStep{8}\\%
                   (5) & (6) & (7) & (8) \\
                   \PicStep{9}&%
                   \PicStep{10}&%
                   \PicStep{11}&%
                   \PicStep{12}\\%
                   (9) & (10) & (11) & (12) \\
                   \PicStep{13}&%
                   \PicStep{14}&%
                   \PicStep{15}&%
                   \PicStep{16}\\%
                   (13) & (14) & (15) & (16) \\
               \end{tabular}
               \captionof{figure}{%
                  An example of the execution of the algorithm
                  \AlgGPermutNN of \secref{g:p:alg}. %
               }
               \figlab{g:p:execution}
           \end{minipage}
       \end{minipage}
    }
\end{figure}

\subsection{Analysis}

Let $\OptSet = \brc{\ocen_1, \ldots, \ocen_k}$ be an optimal set of
$k$ centers of $\PntSet$. Formally, it is a set of $k$
points in $\PntSet$ that minimizes the quantity $\roptX{k} =
\max_{\query \in \PntSet} \distSet{\query}{\OptSet}$. Specifically,
$\roptX{k}$ is the smallest possible radius such that $k$ closed balls
of that radius centered at points in $\PntSet$, cover $\PntSet$.  Our
claim is that after $O(k)$ iterations of the algorithm \AlgGPermutNN,
the sequence of points provides a similar quality clustering of
$\PntSet$.

For any given point $\pnt \in \Re^d$ we can cover the sphere of
directions centered at $\pnt$ by narrow cones of angular diameter at
most $\pi/12$. We fix such a covering, denoting the set of cones by
$\ConeSet_\pnt$, and observe that the number of such cones is a
constant $\cconst_d$ that depends on the dimension. Moreover, by
simple translation we can transfer such a covering to be centered at
any point $\pnt' \in \Re^d$.

% We can now formalize our main result for \AlgGPermutNN.

\begin{lemma}
    \lemlab{only:one:in:cone}%
    After $\mu = k \cconst_d$ iterations, for any optimal center
    $\ocen_i \in \OptSet$, we have $\distSet{\ocen_i}{\GSet_\mu} \leq
    3\ropt$, where $\ropt = \roptX{k}$.
\end{lemma}

\begin{proof}
    If for any $j \leq \mu$, we have $r_j \leq 3\ropt$ then all the
    points of $\Domain_{j-1} \supseteq \PntSet$ are in distance at
    most $3\ropt$ from $\GSet_j$, and the claim trivially holds as
    $\OptSet \subseteq \PntSet$.
    
    Let $\ocen$ be an optimal center and let $\PntSet_\ocen$ be the
    set of points of $\PntSet$ that are closest to $\ocen$ among all
    the centers of $\OptSet$, i.e., $\PntSet_\ocen$ is the cluster of
    $\ocen$ in the optimal clustering.  Fix a cone $\cone$ from
    $\ConeSet_\ocen$ ($\cone$'s apex is at $\ocen$). Consider the
    output sequence $\nnp_0, \nnp_1, \ldots$, and the corresponding
    query sequence $\query_0, \query_1, \ldots$ computed by the
    algorithm. In the following, we use the property of the algorithm
    that $r_1 \geq r_2 \geq \cdots$, where
    $r_i = \distSet{\query_i}{\GSet_{i-1}}$.  A point $\query_j$ is
    \emphi{admissible} if
    \begin{inparaenum}[(i)]
        \item $\nnp_j \in \PntSet_\ocen$, and
        \item $\query_j \in \cone$ (in particular, $\nnp_j$ is not
        necessarily in $\cone$).
    \end{inparaenum}
    
    \parpic[r]{{\includegraphics{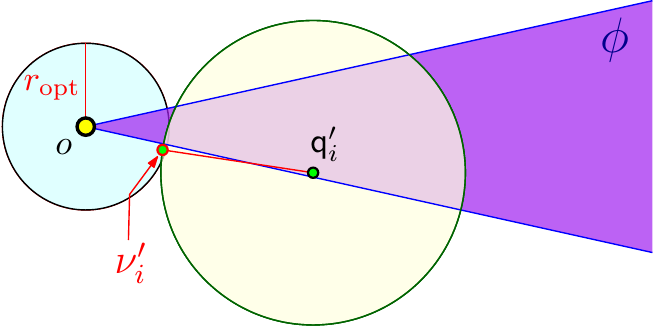}}}
    
    We proceed to show that there are at most $O(1)$ admissible points
    for a fixed cone, which by a packing argument will imply the claim
    as every $\query_j$ is admissible for exactly one cone.  Consider
    the induced subsequence of the output sequence restricted to the
    admissible points of $\phi$: $\nnp_1', \nnp_2', \ldots$, and let
    $\query_1', \query_2'$, $\ldots$ be the corresponding query points
    used by the algorithm.  Formally, for a point $\nnp_i'$ in this
    sequence, let $\iterX{i}$ be the iteration of the algorithm it was
    created.  Thus, for all $i$, we have
    $\query_i' = \query_{\iterX{i}}$ and $\nnp_i' = \nnp_{\iterX{i}}$.

%    and $\GSet_{i}' = \brc{\nnp_1,\ldots, \nnp_{\iterX{i}}}$.

    \smallskip%
    
    Observe that
    $\PntSet_\ocen \subseteq \PntSet \cap \ballY{\ocen}{\ropt}$. This
    implies that %
    \InSoCGVer{%
       \begin{math}
           \distY{\nnp_j'}{\ocen} \leq \ropt,
       \end{math} %
       for all $j$. %
    }%
    \InNotSoCGVer{%
       \begin{align*}
           \distY{\nnp_j'}{\ocen} \leq \ropt, \qquad \text{ for all }
           j.
       \end{align*}%
    }
    
    Let $\ell_i' = \distY{\query_i'}{\nnp_i'}$ and
    $r_i' = \distSet{\query_i'}{\GSet_{\iterX{i} -1}}$.
    Observe that for $i > 1$, we have $\ell_i' \leq r_i' \leq
    \ell_i' + 2 \ropt$, as $\nnp_{i-1}' \in \PntSet_\ocen$. Hence,
    if $\ell_i' \leq \ropt$, then $r_i' \leq 3\ropt$, and we are done.
    This implies that for any $i,j$, such that $1 < i<j$, it must be
    that $\distY{\query_i'}{\query_j'} \geq \ell_i' > \ropt$, as the
    algorithm carves out a ball of radius $\ell_i'$ around
    $\query_i'$, and $\query_j'$ must be outside this ball.

    By a standard packing argument, there can be only $O(1)$ points in
    the sequence $\query_2', \query_3', \ldots$ that are within
    distance at most $10\ropt$ from $\ocen$. If there are no points
    beyond this distance, we are done. Otherwise, let $i>1$ be the
    minimum index, such that $\query_i'$ is at distance larger than
    $10\ropt$ from $\ocen$. We now prove that the points of $\cone \setminus
    \ballY{\query_i'}{\ell_i'}$ are of two types --- those contained
    within $\ballY{\ocen}{3\ropt}$ and those that lie at distance
    greater than $(4/3)\ell_i'$ from $\ocen$. %

    \parpic[r]{\includegraphics{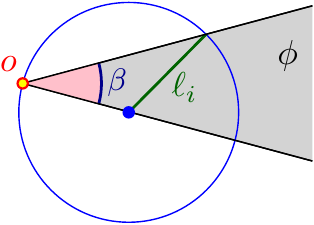}} 

    To see this, observe that since the angle of the cone was chosen
    to be sufficiently small, $\ballY{\query_i'}{\ell_i'}$ splits
    $\cone$ into two components, where all the points in the component
    containing $\ocen$ are distance $<3 \ropt$ from $\ocen$.  The
    minimum distance to $\ocen$ (from a point in the component not 
    containing $\ocen$) is realized when $\query_i'$ is on the boundary
    of $\cone$ and $\ocen$ is on the boundary of
    $\ballY{\query_i'}{\ell_i'}$.  Then the distance of any point of
    $\cone \setminus \ballY{\query_i'}{\ell_i'}$ from $\ocen$ is at
    least
    \begin{math}
        2 \ell_i'\cos(\beta)%
        \geq%
        2 \ell_i' \sqrt{3/4}%
        \geq 1.73 \ell_i,
    \end{math}
    as the opening angle of the cone is at most $\pi/12$. See figure
    on the right. The general case is somewhat more complicated as
    $\ocen$ might be in distance at most $\ropt$ from the boundary of
    $\ballY{\query_i'}{\ell_i'}$, but as $\ell_i \geq 10\ropt$, the
    claim still holds --- we omit the tedious but straightforward
    calculations.

    In particular, this implies that any later point $\query_k'$ in
    the sequence (i.e., $k > i$) is either one of the $O(1)$ close
    points, or it must be far away, but then it is easy to argue that
    $r_k'$ must be larger than $r_i'$, which is a contradiction as
    $r_2\geq r_3 \geq \cdots$ (as $r_i'$ appears before $r_k'$ in this
    sequence).
\end{proof}

The above lemma readily implies the following.
% \lemlab{only:one:in:cone}%

\begin{theorem}
    \thmlab{g:p:via:nn}%
    Let $\PntSet \subseteq \Domain$ be a given set of points in
    $\Re^d$ (not necessarily finite), where $\Domain$ is a 
    %convex
    bounded set in $\Re^d$. Furthermore, assume that $\PntSet$ can be
    accessed only via a data structure $\BNN$ that answers exact
    nearest-neighbor (\NN) queries on $\PntSet$.  The algorithm
    \AlgGPermutNN, described in \secref{g:p:alg}, computes a
    permutation $\permut{\nnp_0, \ldots}$ of $\PntSet$, such that, for
    any $k>0$,
    \begin{math}
        \PntSet \subseteq \bigcup_{i=}^{c k} \ballY{
           \nnp_i}{\roptX{k}},
    \end{math}
    where $c$ is a constant (independent of $k$), and $\roptX{k}$ is
    the minimum radius of $k$ balls (of the same radius) needed to
    cover $\PntSet$.
    
    The algorithm can be implemented, such that running it for $i$
    iterations, takes polynomial time in $i$ and involves $i$ calls to
    $\BNN$.
\end{theorem}
\begin{proof}
    Using \lemref{greedy:clustering} \itemref{packing} in
    \lemref{only:one:in:cone} implies the result. As for the running
    time, naively one needs to maintain the arrangement of balls
    inside the domain, and this can be done in polynomial time in the
    number of balls.
\end{proof}

\begin{observation}
    If $\PntSet$ is finite of size $n$, the above theorem implies that
    after $i \geq cn$ iterations, one can recover the entire point set
    $\PntSet$ (as $\roptX{n}=0$). Therefore $cn$ is an upper bound on
    the number of queries for any problem. Note however that in
    general our goal is to demonstrate when problems can be solved
    using a significantly smaller amount of \NN queries. 
    % In particular, the above algorithm does not
    % require knowing the cardinality of $\PntSet$.
\end{observation}

The above also implies an algorithm for approximating the diameter.
\begin{lemma}
    \lemlab{diam:const:apprx}%
    \RefProofInAppendix{d:c:a} %
    %a
    Consider the setting of \thmref{g:p:via:nn} using an exact
    nearest-neighbor oracle. Suppose that the algorithm is run for $m
    = \cconst_d + 1$ iterations, and let $\nnp_1, \ldots, \nnp_m$ be
    the set of output centers and $r_1, \ldots, r_m$ be the
    corresponding distances.  Then, $\diameterX{\PntSet}/3 \leq \max
    (\diameterX{\nnp_1,\ldots,\nnp_m}, r_m ) \leq 3 \cdot
    \diameterX{\PntSet}$.
\end{lemma}

\begin{proof:in:appendix:e}{\lemref{diam:const:apprx}}{d:c:a}
    Since the discrete one-center clustering radius lies in the
    interval $[\diameterX{\PntSet}/2, \diameterX{\PntSet}]$,
    \lemref{only:one:in:cone} implies that $r_m \leq 3 \ropt \leq 3
    \cdot \diameterX{\PntSet}$.  Moreover, each $\nnp_i$ is in
    $\PntSet$, and so $\diameterX{\nnp_1, \ldots, \nnp_m} \leq
    \diameterX{\PntSet}$. Thus the upper bound follows. 

    For the lower bound, observe that if
    \begin{math}
        \diameterX{\nnp_1, \ldots, \nnp_m} < \diameterX{\PntSet}/3,
    \end{math}
    as well as $r_m < \diameterX{\PntSet}/3$, then it must be true
    that
    $\PntSet \subseteq \Domain_{m - 1} \subseteq \bigcup_{j = 1}^l
    \ballY{\nnp_j}{r_m}$
    has diameter less than $\diameterX{\PntSet}$, a contradiction.
\end{proof:in:appendix:e}

\subsection{Using approximate nearest-neighbor search}

If we are using an \ANN black box $\BANN$ to implement the algorithm,
one can no longer scoop away the ball $\ball_i = \ballY{\query_i}{
   \distY{\query_i}{\nnp_i}}$ at the $i$\th iteration, as it might
contain some of the points of $\PntSet$. Instead, one has to be more
conservative, and use the ball
\begin{math}
    \ball_i' = \ballY{\query_i}{ (1-\eps) \distY{\query_i}{\nnp_i}}
\end{math}
Now, we might need to perform several queries till the volume being
scooped away is equivalent to a single exact query.

Specifically, let $\PntSet$ be a finite set, and consider its
associated \emphi{spread}:
\begin{align*}
    \Spread = \frac{\diameterX{\Domain_0}}%
    {\min_{\pnt, \pntA \in \PntSet} \distY{\pnt}{\pntA}}.
\end{align*}
We can no longer claim, as in \lemref{only:one:in:cone}, that each
cone would be visited only one time (or constant number of times).
Instead, it is easy to verify that each query point in the cone,
shrinks the diameter of the domain restricted to the cone by a factor
of roughly $\eps$.  As such, at most
$O\pth{ \log_{1/\eps} \Spread } = O\pth{ \eps^{-1} \log \Spread }$
query points would be associated with each cone.

\begin{corollary}
    Consider the setting of \thmref{g:p:via:nn}, with the modification
    that we use a $(1+\eps)$-\ANN data structure $\BANN$ to access
    $\PntSet$. Then, for any $k$,
    \begin{math}
        \PntSet \subseteq \bigcup_{i=1}^{f(k)} \ballY{
           \nnp_i}{\roptX{k}},
    \end{math}
    where $f(k) = O\pth{ k \eps^{-1} \log \Spread}$.
\end{corollary}

%----------------------------------
\subsection{Discussion}

\myparagraph{Outer approximation.} %
As implied by the algorithm description, one can think about the
algorithm providing an outer approximation to the set:
$\Domain_1\supseteq \Domain_2 \supseteq \cdots \supseteq \PntSet$. As
demonstrated in \figref{g:p:execution}, the sequence of points
computed by the algorithm seems to be a reasonable greedy permutation
of the underlying set. However, the generated outer approximation
seems to be inferior. If the purpose is to obtain a better outer
approximation, a better strategy may be to pick the $i$\th query point
$\query_i$ as the point inside $\Domain_i$ farthest away from
$\bd\Domain_{i-1} \cup \GSet_{i-1}$

\myparagraph{Implementation details.} %
We have not spent any effort to describe in detail the algorithm of
\thmref{g:p:via:nn}, mainly because an implementation of the exact
version seems quite challenging in practice. A more practical approach
would be to describe the uncovered domain $\Domain_i$ approximately,
by approximating from the inside, every ball $\ball_i$ by an
$O\pth{1/\eps^d}$ grid of cubes, and maintaining these cubes using a
(compressed) quadtree. This provides an explicit representation of
the complement of the union of the approximate balls. Next, one would
need to maintain for every free leaf of this quadtree, a list of points of
$\GSet_i$ that might serve as its nearest neighbors --- in the spirit
of approximate Voronoi diagrams \cite{h-gaa-11}.

% ------------------------------------------------------------------
% ------------------------------------------------------------------

\section{Convex-hull membership queries via %
   proximity queries}
\seclab{ch:memb}

Let $\PntSet$ be a set of $n$ points in $\Re^d$, let $\diam$ denote
$\PntSet$'s diameter, and let $\eps > 0$ be a prespecified parameter.
We assume that the value of $\diam$ is known, although a constant
approximation to this value is sufficient for our purposes. (See
\lemref{diam:const:apprx} on how to compute this under reasonable
assumptions.)

Let $\CH = \CHX{\PntSet}$ denote $\PntSet$'s convex hull. Given a
query point $\query \in \Re^d$, the task at hand is to decide if
$\query$ is in $\CH$.  As before, we assume that our only access to
$\PntSet$ is via an \ANN data structure. There are two possible 
outputs:
\begin{compactenum}[\qquad(A)]
    \item \RetInHull: if $\query \in \CH$, and
    \item \RetOutHull: if $\query$ is at distance greater than $\eps \diam$ from
    $\CH$,
\end{compactenum}
Either answer is acceptable if $\query$ lies within distance 
$\eps \diam$ of $\partial \CH$.

\subsection{Convex hull membership queries using %
   exact extremal queries}
   
   We first solve the problem using exact extremal queries and then 
later show these queries can be answered approximately with \ANN queries.

\subsubsection{The algorithm}
\seclab{e:extremal}

\parpic[r]{%
   \begin{minipage}{0.3\linewidth}%
       \vspace{-1cm}
       \includegraphics{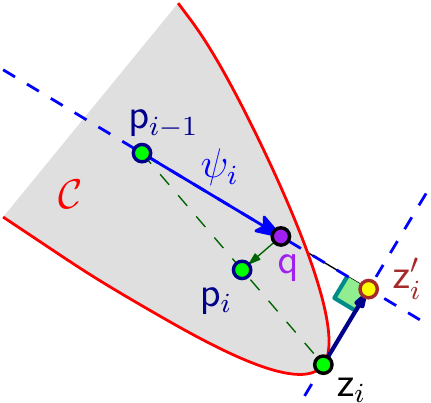}%
   \end{minipage}%
}% \vspace{1cm}

\noindent%
We construct a sequence of points $\pnt_0, \pnt_1, \ldots$ each
guaranteed to be in the convex hull $\CH$ of $\PntSet$ and use them to
determine whether $\query \in \CH$.  The algorithm is as follows.
Let $\pnt_0$ be an arbitrary point of $\PntSet$.  For $i > 0$, in the
$i$\th iteration, the algorithm checks whether
$\distY{\pnt_{i-1}}{\query} \leq \eps \diam$, and if so the algorithm
outputs \RetInHull and stops.

Otherwise, consider the ray $\ray_{i}$ emanating from $\pnt_{i-1}$ in
the direction of $\query$.  The algorithm computes the point 
$\pntB_i \in \PntSet$ that is extremal in the direction of this ray.
If the projection $\pntB_i'$ of $\pntB_i$ on the line supporting
$\ray_i$ is between $\pnt_{i-1}$ and $\query$, then $\query$ is
outside the convex-hull $\CH$, and the algorithm stops and returns
\RetOutHull.  Otherwise, the algorithm sets $\pnt_i$ to be the
projection of $\query$ on the line segment $\pnt_{i-1} \pntB_i$, and
continues to the next iteration. See figure on the right, and
\figref{direct}.

The algorithm performs $O(1/\eps^2)$ iterations, and returns
\RetOutHull if it did not stop earlier.

% The algorithm performs $N = 2\log_{1/(1-\eps^2)}(1/\eps)$
% iterations, and returns \RetOutHull if it did not stop earlier.

\subsubsection{Analysis}

\begin{lemma}
    \lemlab{shrink}%
    If the algorithm runs for more than $i$ iterations, then
    $\DD_i < \pth{1-\frac{\eps^2}{2}} \DD_{i-1}$, where
    $\DD_i = \distY{\query}{\pnt_i}$.
\end{lemma}

\begin{proof}
    By construction, $\pnt_i$, $\pnt_{i-1}$, and $\query$ form a right
    angle triangle. The proof now follows by a direct trigonometric
    argument.  Consider \figref{direct}.  We have the following
    properties:
    
    \noindent
    \begin{minipage}{0.23\linewidth}
        \hfill%
        {\includegraphics[page=2]{figs/ray}}
        \captionof{figure}{}
        \figlab{direct}
    \end{minipage}%
    \begin{minipage}{0.76\linewidth}
        \smallskip
        \begin{compactenum}[\quad (A)]
            \item The triangles $\triangle \pnt_{i-1} \pntB_i'
            \pntB_i$ and $\triangle \pnt_{i-1} \pnt_i \query$ are
            similar.

            \item Because the algorithm has not terminated in the $i$\th
            iteration, $\distY{\pnt_{i-1}}{\query} > \eps \diam$.
            
            \item The point $\query$ must be between $\pnt_{i-1}$ and
            $\pntB_i'$, as otherwise the algorithm would have
            terminated.
            % As such, $\distY{\pntB_i'}{\pntB_i} \geq
            % \distY{\pnt_i}{\query} > \eps \diam$.
            Thus, $\distY{\pnt_{i-1}}{\pntB_i'} \geq
            \distY{\pnt_{i-1}}{\query} > \eps \diam $.
            
            \item We have $\distY{\pnt_{i-1}}{\pntB_i}\leq \diam$,
            since both points are in $\CH$.
        \end{compactenum}
    \end{minipage}

    % \hrule
    
    \smallskip\noindent%
    We conclude that
    \begin{math}
        \ds%
        \cos \beta%
        = %
        \frac{\distY{\pnt_{i-1}}{\pntB_i'}}{
           \distY{\pnt_{i-1}}{\pntB_i} }%
        > %
        \frac{\eps \diam}{\diam}%
        =%
        \eps.
    \end{math}
    Now, we have
    \begin{align*}
        \distY{\query}{\pnt_i} %
        &= %
        \distY{\query}{\pnt_{i-1}} \sin \beta%
        = %
        \distY{\query}{\pnt_{i-1}} \sqrt{ 1 -\cos^2 \beta}%
        <%
        \sqrt{1-\eps^2} \distY{\query}{\pnt_{i-1}}%
        \\&%
        <%
        \pth{1-\frac{\eps^2}{2}} \distY{\query}{\pnt_{i-1}},
    \end{align*}
    since $(1-\eps^2/2)^2 > 1-\eps^2$.
\end{proof}

\begin{lemma}
    Either the algorithm stops within $O\pth{1/\eps^2}$ iterations
    with a correct answer, or the query point is more than $\eps
    \diam$ far from the convex hull $\CH$; in the latter case, since
    the algorithm says \RetOutHull its output is correct.
\end{lemma}
\begin{proof}
    If the algorithm stops before it completes the maximum number of
    iterations, it can be verified that the output is correct as there
    is an easy certificate for this in each of the possible cases.
    
    Otherwise, suppose that the query point is within $\eps \diam$ of
    $\CH$. We argue that this leads to a contradiction; thus the query
    point must be more than $\eps \diam$ far from $\CH$ and the output
    of the algorithm is correct.  Observe that $\DD_i$ is a monotone
    decreasing quantity that starts at values $\leq \diam$ (i.e,
    $\DD_0 \leq \diam$), since otherwise the algorithm terminates after
    the first iteration, as $\pntB_1'$ would be between $\query$ and
    $\pnt_0$ on $\ray_1$.
    
    Consider the $j$\th \emphi{epoch} to be block of iterations of the
    algorithm, where $2^{-j} \Delta < \DD_i \leq 2^{-j +
       1}\Delta$. Following the proof of \lemref{shrink}, one observes
    that during the $j$\th epoch one can set $\eps_j = 1/2^j$ in place
    of $\eps$, and using the argument it is easy to show that the
    $j$\th epoch lasts $O(1/\eps_j^2)$ iterations. By assumption,
    since the algorithm continued for the maximum number of iterations
    we have $\DD_i > \eps \diam$, and so the maximum number
    of epochs is $\ceil{\lg( 1/\eps)}$.  As such, the total number of
    iterations is $\sum_{j=1}^{\ceil{\lg (1/\eps)}} O(1/\eps_j^2) = O(
    1/\eps^2)$. Since the algorithm did not stop this is a
    contradiction.
\end{proof}

\subsubsection{Approximate extremal queries}

\noindent%
\begin{minipage}{0.65\linewidth}
    For our purposes, approximate extremal queries on $\PntSet$ are
    sufficient.

\begin{defn}
    A data structure provides \emphi{$\eps$-approximate extremal
       queries} for $\PntSet$, if for any query unit vector $\vVec$,
    it returns a point $\pnt$, such that
    \begin{align*}
        \forall \pntA \in \PntSet, \qquad%
        \DotProd{\vVec}{\pntA} \leq \DotProd{\vVec}{\pnt} + \eps \cdot
        \diameterX{\PntSet},
    \end{align*}
    where $\DotProd{\vVec}{\pntA}$ denotes the dot-product of $\vVec$
    with $\pntA$.
\end{defn}
\end{minipage}%
%\begin{minipage}{0.8\linewidth}
%    \parpic[r]{%
\begin{minipage}{0.35\linewidth}%
    \quad\includegraphics{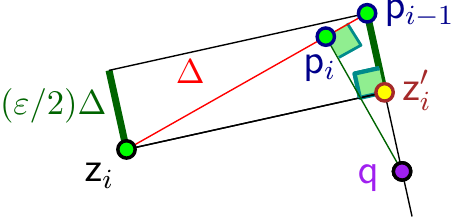}%
    \captionof{figure}{Worse case if extremal queries are
       approximate.}%
    \figlab{extremal:2}%
\end{minipage}%
% }%

\bigskip

One can now modify the algorithm of \secref{e:extremal} to use, say,
$\eps/4$-\si{appro\-\si{ximate}} extremal queries on
$\PntSet$. Indeed, one modifies the algorithm so it stops only if
$\pntB_i$ is on the segment $\pnt_{i-1} \query$, and it is in distance
more than $\eps \diam /4$ away from $\query$. Otherwise the algorithm
continues. It is straightforward but tedious to prove that the same
algorithm performs asymptotically the same number of iterations
(intuitively, all that happens is that the constants get slightly
worse). The worse case as far progress in a single iteration is
depicted in \figref{extremal:2}.

\begin{lemma}
    The algorithm of \secref{e:extremal} can be modified to use
    $\eps/4$-approximate extremal queries and output a correct answer
    after performing $O\pth{1/\eps^2}$ iterations.
\end{lemma}

\subsection{Convex-hull membership via \ANN queries}

\subsubsection{Approximate extremal queries via \ANN queries}

The basic idea is to replace the extremal empty half-space query, by
an \ANN query. Specifically, a $(1+\delta)$-\ANN query performed at
$\query$ returns us a point $\pnt$, such that
\begin{align*}
    \forall \pntA \in \PntSet,%
    \qquad%
    \distY{\query}{\pnt}\leq (1+\delta) \distY{\query}{\pntA}.
\end{align*}
Namely, $\ballY{\query}{\frac{\distY{\query}{\pnt}}{1+\delta}}$ does
not contain any points of $\PntSet$. Locally, a ball looks like a
halfspace, and so by taking the query point to be sufficiently far and
the approximation parameter to be sufficiently small, the resulting
empty ball and its associated \ANN can be used as the answer to an
extremal direction query.

\subsubsection{The modified algorithm}
Assume the algorithm is given a data structure $\BANN$ that can answer
$(1+\delta)$-\ANN queries on $\PntSet$. Also assume that it is
provided with an initial point $\pnt_0 \in \PntSet$, and a value
$\aDiam$ that is, say, a 2-approximation to $\diam =
\diameterX{\PntSet}$, that is $\diam \leq \aDiam \leq 2\diam$.

In the $i$\th iteration, the algorithm considers (again) the ray
$\ray_i$ starting from $\pnt_i$, in the direction of $\query$. Let
$\query_i$ be the point within distance, say,
\begin{align}
    \tau = c\aDiam/\eps
    \eqlab{tau:def}%
\end{align}
from $\pnt_{i-1}$ along $\ray_i$, where $c$ is an appropriate constant
to be determined shortly. Next, let $\pntB_i$ be the $(1+\delta)$-\ANN
returned by $\BANN$ for the query point $\query_i$, where the value of
$\delta$ would be specified shortly. The algorithm now continues as
before, by setting $\pnt_i$ to be the nearest point on $\pnt_{i-1}
\pntB_i$ to $\query$. Naturally, if $\distY{\query}{\pnt_i}$ falls
below $\eps \aDiam/2$, the algorithm stops, and returns \RetInHull,
and otherwise the algorithm continues to the next iteration. As
before, if the algorithm reaches the $N$\th iteration, it stops and
returns \RetOutHull, where $N = O(1/\eps^2)$.

\subsubsection{Analysis}

\begin{lemma}%
    \lemlab{a:n:n:extremal}%
    Let $0 < \eps \leq 1$ be a prespecified parameter, and let $\delta
    = \eps^2/(32 -\eps)^2 = O(\eps^2)$. Then, a $(1+\delta)$-\ANN
    query done using $\query_i$ performed by the algorithm, returns a
    point $\pntB_i$ which is a valid $\eps$-approximate extremal query
    on $\PntSet$, in the direction of $\ray_i$.
\end{lemma}
\begin{proof}
    \begin{figure}[t]
        \centerline{\includegraphics{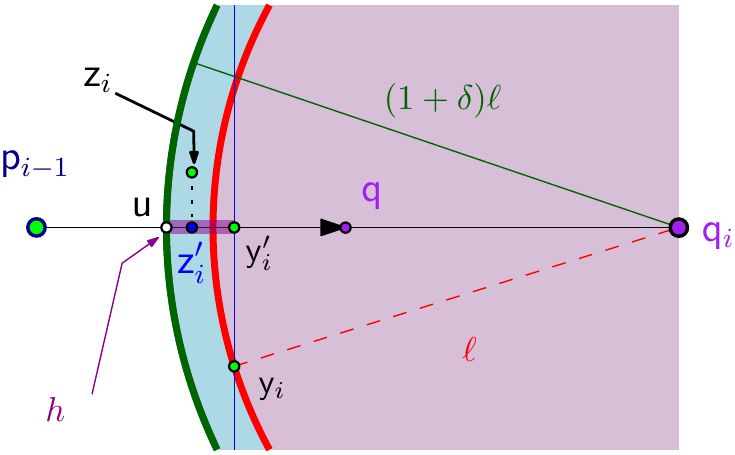}}
        \caption{Illustration of the proof of
           \lemref{a:n:n:extremal}.}%
        \figlab{a:n:n:e}%
    \end{figure}
    Consider the extreme point $\pntC_i \in \PntSet$ in the direction
    of $\ray_i$.  Let $\pntC_i'$ be the projection of $\pntC_i$ to the
    segment $\pnt_{i-1}\query_i$, and let
    $\ell = \distY{\query_i}{\pntC_i}$. See \figref{a:n:n:e}.

    The $(1+\delta)$-\ANN to $\query_i$ (i.e., the point $\pntB_i$),
    must be inside the ball $\ball =
    \ballY{\query_i}{(1+\delta)\ell}$, and let $\pntB_i'$ be its
    projection to the segment $\pnt_{i-1}\query_i$.
    
    Now, if we interpret $\pntB_i$ as the returned answer for
    the approximate extremal query, then the error is the distance
    $\distY{\pntB_i'}{\pntC_i'}$, which is maximized if $\pntB_i'$ is
    as close to $\pnt_{i-1}$ as possible. In particular, let $\pntD$
    be the point in distance $(1+\delta)\ell$ from $\query_i$ along
    the segment $\pnt_{i-1}\query_i$. We then have that
    \begin{math}
        \distY{\pntB_i'}{\pntC_i'}%
        \leq %
        h = \distY{\pntD}{\pntC_i'}.
    \end{math}
    Now, since $\distY{\pntC_i'}{\pntC_i} \leq
    \distY{\pnt_{i-1}}{\pntC_i} \leq \aDiam$, we have
    \begin{align*}
        h%
        & =%
        \distY{\pntD}{\pntC_i'}%
        \leq%
        (1+\delta) \ell - \distY{\pntC_i'}{\query_i}%
        = %
        (1+\delta) \ell - \sqrt{\ell^2 - \distY{\pntC_i'}{\pntC_i}^2
        }%
        \\ &%
        \leq%
        (1+\delta) \ell - \sqrt{\ell^2 - \pth{\aDiam}^2 }%
        = %
        \frac{(1+\delta)^2 \ell^2 - \ell^2 + \pth{\aDiam}^2 }
        {(1+\delta) \ell + \sqrt{\ell^2 - \pth{\aDiam}^2 }}%
        \leq %
        \frac{ (2\delta + \delta^2)\ell^2 +
           \pth{\sqrt{\delta}\ell}^2 } {\ell }%
        \\&%
        \leq%
        \frac{ 4\delta \ell^2 } {\ell }%
        =%
        4\delta \ell,
    \end{align*}
    since $\delta \leq 1$, and assuming that $ \aDiam \leq
    \sqrt{\delta} \ell$. For our purposes, we need that
    % $4\delta \ell \leq \eps \diam /2$,
    $4 \delta \ell \leq \eps \diam$. Both of these constraints
    translate to the inequalities,
    \begin{math}
        \ds%
        \pth{\frac{\aDiam}{\ell}}^2%
        \leq%
        \delta%
        \leq%
        \frac{\eps \diam}{4\ell }.
    \end{math}
    Observe that, by the triangle inequality, it follows that
    \begin{align*}
        \ell = \distY{\query_i}{\pntC_i}%
        \leq%
        \distY{\query_i}{\pnt_{i-1}} + \distY{\pnt_{i-1}}{\pntC_i}
        \leq%
        \tau + \diam.
    \end{align*}
    A similar argument implies that $\ell \geq \tau - \diam$.  In
    particular, it is enough to satisfy the constraint
    \begin{math}
        \pth{\frac{\aDiam}{\tau - \diam}}^2%
        \leq%
        \delta%
        \leq%
        \frac{\eps \diam}{4(\tau + \diam) },
    \end{math}
    which is satisfied if
    \begin{math}
        % \ds%
        % \\
        %
        \pth{\frac{\aDiam}{\tau - \diam'}}^2%
        \leq%
        \delta%
        \leq%
        \frac{\eps \diam'/2}{4(\tau + \diam') },
    \end{math}
    as $\diam \leq \aDiam\leq 2\diam$.  Substituting the value of
    $\tau = c\aDiam/\eps$, see \Eqref{tau:def}, this is equivalent to
    \begin{math}
        \pth{\frac{1}{c/\eps - 1}}^2%
        \leq%
        \delta%
        \leq%
        \frac{\eps /2}{4(c/\eps + 1) },
    \end{math}
    which holds for $c = 32$, as can be easily verified, and setting
    $\delta = \eps^2/(32 -\eps)^2 = O( \eps^2)$.
\end{proof}

\begin{theorem}
    Given a set $\PntSet$ of $n$ points in $\Re^d$, let $\eps \in
    (0,1]$ be a parameter, and let $\aDiam$ be a constant
    approximation to the diameter of $\PntSet$. Assume that you are
    given a data structure that can answer $(1+\delta)$-\ANN queries
    on $\PntSet$, for $\delta = O(\eps^2)$. Then, given a query point
    $\query$, one can decide, by performing $O(1/\eps^2)$
    $(1+\delta)$-\ANN queries whether $\query$ is inside the
    convex-hull $\CH = \CHX{\PntSet}$. Specifically, the algorithm
    returns
    \begin{compactitem}
        \item \RetInHull: if $\query \in \CH$, and
        \item \RetOutHull: if $\query$ is more than $\eps \diam$ away
        from $\CH$, where $\diam = \diameterX{\PntSet}$.
    \end{compactitem}
    The algorithm is allowed to return either answer if $\query \notin
    \CH$, but the distance of $\query$ from $\CH$ is at most $\eps
    \diam$.
\end{theorem}

\section{Density clustering}
\seclab{density:clust}

\subsection{Definition}

Given a set $\PntSet$ of $n$ points in $\Re^d$, and a parameter $k$
with $1 \leq k \leq n$, we are interested in computing a set $\CenSet
\subseteq \PntSet$ of ``centers'', such that each center is assigned
at most $k$ points, and the number of centers is (roughly) $n/k$. In
addition, we require that:
\smallskip%
\begin{compactenum}[\qquad(A)]
    \item A point of $\PntSet$ is assigned to its nearest neighbor in
    $\CenSet$ (i.e., $\CenSet$ induces a \emphi{Voronoi partition} of
    $\PntSet$).
    
    \item The centers come from the original point set.
\end{compactenum}
\smallskip%
Intuitively, this clustering tries to capture the local density --- in
areas where the density is low, the clusters can be quite large (in
the volume they occupy), but in regions with high density the clusters
have to be tight and relatively ``small''.

Formally, given a set of centers $\CenSet$, and a center $\cen \in
\CenSet$, its \emphi{cluster} is
\begin{align*}
    \PntSet_\cen = \Set{ \pnt \in \PntSet}{ \bigl. \distY{\cen}{\pnt}
       < \distSet{\pnt}{\CenSet \setminus \brc{c} \bigr.} \Bigl.  },%
\end{align*}
where $\distSet{\cen}{X} = \min_{\pnt \in X} \distY{\cen}{\pnt}$ (and
assuming for the sake of simplicity of exposition that all distances
are distinct). The resulting \emphi{clustering} is
$\PartitionY{\PntSet}{\CenSet} = \brc{ \PntSet_\cen \sep{\cen \in
      \CenSet}}$.  A set of points $\PntSet$, and a set of centers
$\CenSet \subseteq \PntSet$ is a \emphi{$k$-density clustering} of
$\PntSet$ if for any $\cen \in \CenSet$, we have
$\cardin{\PntSet_\cen} \leq k$. As mentioned, we want to compute a
balanced partitioning, i.e., one where the number of centers is
roughly $n / k$. We show below that this is not always possible in
high enough dimensions.

% The purpose is to compute such a $k$-density clustering such that

\subsubsection{A counterexample in high dimension}

\begin{lemma}
    For any integer $n > 0$, there exists a set $\PntSet$ of $n$
    points in $\Re^n$, such that for any $k < n$, a $k$-density
    clustering of $\PntSet$ must use at least $n-k+1$ centers.
\end{lemma}

\begin{proof}
    Let $n$ be a parameter. For $i=1, \ldots, n$, let $\ell_i =
    \sqrt{1- 2^{-i-1}}$, and let $\pnt_i$ be a point of $\Re^n$ that
    has zero in all coordinates except he $i$\th one, where its value
    is $\ell_i$.  Let $\PntSet = \brc{ \pnt_1, \ldots, \pnt_n}
    \subseteq \Re^n$.  Now, $d_{i,j} = \distY{\pnt_i}{\pnt_j} = \sqrt{
       \ell_i^2 + \ell_j^2} = \sqrt{ 2 - 2^{-i-1} - 2^{-j-1}}$. For $i
    <j$ and $i' < j'$, we have
    \begin{align*}
        d_{i,j} < d_{i',j'}%
        &\iff%
        2 - 2^{-i-1} - 2^{-j-1} < 2 - 2^{-i'-1} - 2^{-j'-1}%
        \iff%
        2^{-i'} + 2^{-j'} < 2^{-i} + 2^{-j}\\
        &\iff%
        \begin{cases}
            i=i' \text{ and } j < j', or\\
            i<i'.
        \end{cases}
    \end{align*}
    That is, the distance of the $i$\th point to all the following
    points $\SuffixSet_{i+1} = \brc{\pnt_{i+1}, \ldots, \pnt_n}$ is
    smaller than the distance between any pair of points of
    $\SuffixSet_{i+1}$.
    
    Now, consider any set of centers $\CenSet \subseteq \PntSet$, and
    let $\cen = \pnt_i$ be the point with the lowest index that
    belongs to $\CenSet$. Clearly,
    $\PntSet_\cen = \pth{\PntSet \setminus \CenSet} \cup \brc{\cen}$;
    that is, all the non-center points of $\PntSet$, get assigned by
    the clustering to $\cen$, implying what we want to show.
\end{proof}

\subsection{Algorithms}

\subsubsection{Density clustering via nets}

\begin{lemma}
    \lemlab{slices}%
    For any set of $n$ points $\PntSet$ in $\Re^d$, and a parameter $k
    < n$, there exists a $k$-density clustering with $O\pth{
       \frac{n}{k} \log \frac{n}{k}}$ centers (the $O$ notation hides
    constants that depend on $d$).
\end{lemma}
\begin{proof}
    Consider the hypercube $[-1,1]^d$. Cover its outer faces (which
    are $(d-1)$-dimensional hypercubes) by a grid of side length
    $1/3\sqrt{d}$. Consider a cell $C$ in this grid --- it has
    diameter $\leq 1/3$, and it is easy to verify that the cone
    $\coneC = \Set{ t \pnt }{ \pnt \in C, t \geq 0}$ formed by the
    origin and $C$ has angular diameter $< \pi/3$. This results in a
    set $\ConeSet$ of $N = O(d^d)$ cones covering $\Re^d$.
    
    Furthermore, the cone $\coneC$ is formed by the intersection of
    $2(d-1)$ halfspaces. As such, the range space consisting of all 
    translations of $\coneC$ has \VC dimension at most $d' = O(d^2 \log d)$ \cite[Theorem
    5.22]{h-gaa-11}. Let \emphi{slice} be the set that is formed by
    the intersection of such a cone, with a ball centered at its
    apex. The range space of such slices has \VC dimension $d'' = O(
    d+2 + d') = O( d^2 \log d)$, since the \VC dimension of balls in
    $\Re^d$ is $d+2$, and one can combine the range spaces as done
    above, see the book \cite{h-gaa-11} for background on this.
    
    Now, for $\eps = (k/N)/n$, consider an $\eps$-net $\RSample$ of
    the point set $\PntSet$ for slices. The size of such a net is
    $\cardin{\RSample} = O\pth{(d''/\eps) \log \eps^{-1} } = d^{O(d)}
    (n/k) \log (n/k)$.
    % Such a net $\RSample$ is formed by a random
    % sample of size $O\pth{(d''/\eps) \log \eps^{-1} } = d^{O(d)}
    % (n/k)
    % \log (n/k)$ (the probability that this is an $\eps$-net is
    % relatively high, if the sample fails, we can always resample).
    
    Consider a point $\pnt \in \PntSet$ that is in $\RSample$, and
    consider the cones of $\ConeSet$ translated so that their apex is
    $\pnt$. For a cone $\cone \in \ConeSet$, let $\nnp_\cone$ be the
    nearest point to $\pnt$ in the set $\brc{\RSample \setminus
       \brc{\pnt} \Bigl.} \cap \cone$. The key observation is that any
    point in $\PntSet \cap \cone$ that is farther away from $\pnt$
    than $\nnp_\cone$, is closer to $\nnp_\cone$ than to $\pnt$; that
    is, only points closer to $\pnt$ than $\nnp_\cone$ might be
    assigned to $\pnt$ in the Voronoi clustering.  Since
    $\RSample$ is an $\eps$-net for slices, the slice
    $\ballY{\bigl.\pnt}{\distY{\pnt}{\nnp_\cone}} \cap \cone$ contains
    at most $\eps n = k/N$ points of $\PntSet$. It follows that at
    most $k/N$ points of $\PntSet \cap \cone$ are assigned to the
    cluster associated with $\pnt$. By summing over all $N$ cones, 
    at most $(k/N)N = k$ points are assigned to $\pnt$, as desired.
\end{proof}

\subsubsection{The planar case}

\begin{lemma}
    \lemlab{slices:2}%
    \RefProofInAppendix{slices:2} %
    For any set of $n$ points $\PntSet$ in $\Re^2$, and a parameter
    $k$ with $1 \leq k \leq n$, there exists a $k$-density clustering
    with $O\pth{ n/k }$ centers.
\end{lemma}
\begin{proof:in:appendix:e}{\lemref{slices:2}}{slices:2}
    Consider the plane, and a cone $\cone \in \ConeSet$ as used in the
    proof of \lemref{slices}. For a point $\pnt \in \Re^2$, and a
    radius $r$, let $\cone_\pnt$ denote the translated cone having
    $\pnt$ as an apex, and let $s(\pnt, r) = \cone_\pnt \cap
    \ballY{\pnt}{r}$ be a slice induced by this cone. Consider the set
    of all possible slices in the plane:       
    \begin{align*}
        S = \Set{ s(\pnt,r) }{\pnt \in \Re^2, r \geq 0}.
    \end{align*}

    \vspace{-1.2cm}%
    \parpic[r]{\includegraphics{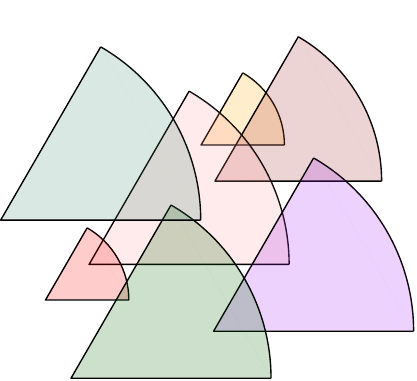}}%
    \vspace{1cm}%

    It is easy to verify that the family $S$ behaves like a system of
    \emph{pseudo-disks}; that is, the boundary of a pair of such
    regions intersects in at most two points, see figure on the
    right. As such, the range space having $\PntSet$ as the ground
    set, and $S$ as the set of possible ranges, has an $\eps$-net of
    size $O(1/\eps)$ \cite{msw-hnlls-90}.  Let $\eps = (k/N)/n$, as in
    the proof of \lemref{slices}, where $N =\cardin{\ConeSet}$.
    Computing such an $\eps$-net, for every cone of $\ConeSet$, and
    taking their union results in the desired set of cluster centers.
\end{proof:in:appendix:e}

% -------------------------------------------------------------------
% -------------------------------------------------------------------
% -------------------------------------------------------------------

\section*{Acknowledgments}
N.K. would like to thank Anil Gannepalli for telling him about Atomic
Force Microscopy.

% -------------------------------------------------------------------------

\InSoCGVer{%
%*flatex input: [./ann_revisit.bbl]
\newcommand{\etalchar}[1]{$^{#1}$}
 \providecommand{\CNFX}[1]{ {\em{\textrm{(#1)}}}}
  \providecommand{\tildegen}{{\protect\raisebox{-0.1cm}{\symbol{'176}\hspace{-0.03cm}}}}
  \providecommand{\SarielWWWPapersAddr}{http://sarielhp.org/p/}
  \providecommand{\SarielWWWPapers}{http://sarielhp.org/p/}
  \providecommand{\urlSarielPaper}[1]{\href{\SarielWWWPapersAddr/#1}{\SarielWWWPapers{}/#1}}
  \providecommand{\Badoiu}{B\u{a}doiu}
  \providecommand{\Barany}{B{\'a}r{\'a}ny}
  \providecommand{\Bronimman}{Br{\"o}nnimann}  \providecommand{\Erdos}{Erd{\H
  o}s}  \providecommand{\Gartner}{G{\"a}rtner}
  \providecommand{\Matousek}{Matou{\v s}ek}
  \providecommand{\Merigot}{M{\'{}e}rigot}
  \providecommand{\CNFSoCG}{\CNFX{SoCG}}
  \providecommand{\CNFCCCG}{\CNFX{CCCG}}
  \providecommand{\CNFFOCS}{\CNFX{FOCS}}
  \providecommand{\CNFSODA}{\CNFX{SODA}}
  \providecommand{\CNFSTOC}{\CNFX{STOC}}
  \providecommand{\CNFBROADNETS}{\CNFX{BROADNETS}}
  \providecommand{\CNFESA}{\CNFX{ESA}}
  \providecommand{\CNFFSTTCS}{\CNFX{FSTTCS}}
  \providecommand{\CNFIJCAI}{\CNFX{IJCAI}}
  \providecommand{\CNFINFOCOM}{\CNFX{INFOCOM}}
  \providecommand{\CNFIPCO}{\CNFX{IPCO}}
  \providecommand{\CNFISAAC}{\CNFX{ISAAC}}
  \providecommand{\CNFLICS}{\CNFX{LICS}}
  \providecommand{\CNFPODS}{\CNFX{PODS}}
  \providecommand{\CNFSWAT}{\CNFX{SWAT}}
  \providecommand{\CNFWADS}{\CNFX{WADS}}

% flatex input end: [./ann_revisit.bbl]
}%
\InNotSoCGVer{%
%*flatex input: [./ann_revisit.bbl]
\newcommand{\etalchar}[1]{$^{#1}$}
 \providecommand{\CNFX}[1]{ {\em{\textrm{(#1)}}}}
  \providecommand{\tildegen}{{\protect\raisebox{-0.1cm}{\symbol{'176}\hspace{-0.03cm}}}}
  \providecommand{\SarielWWWPapersAddr}{http://sarielhp.org/p/}
  \providecommand{\SarielWWWPapers}{http://sarielhp.org/p/}
  \providecommand{\urlSarielPaper}[1]{\href{\SarielWWWPapersAddr/#1}{\SarielWWWPapers{}/#1}}
  \providecommand{\Badoiu}{B\u{a}doiu}
  \providecommand{\Barany}{B{\'a}r{\'a}ny}
  \providecommand{\Bronimman}{Br{\"o}nnimann}  \providecommand{\Erdos}{Erd{\H
  o}s}  \providecommand{\Gartner}{G{\"a}rtner}
  \providecommand{\Matousek}{Matou{\v s}ek}
  \providecommand{\Merigot}{M{\'{}e}rigot}
  \providecommand{\CNFSoCG}{\CNFX{SoCG}}
  \providecommand{\CNFCCCG}{\CNFX{CCCG}}
  \providecommand{\CNFFOCS}{\CNFX{FOCS}}
  \providecommand{\CNFSODA}{\CNFX{SODA}}
  \providecommand{\CNFSTOC}{\CNFX{STOC}}
  \providecommand{\CNFBROADNETS}{\CNFX{BROADNETS}}
  \providecommand{\CNFESA}{\CNFX{ESA}}
  \providecommand{\CNFFSTTCS}{\CNFX{FSTTCS}}
  \providecommand{\CNFIJCAI}{\CNFX{IJCAI}}
  \providecommand{\CNFINFOCOM}{\CNFX{INFOCOM}}
  \providecommand{\CNFIPCO}{\CNFX{IPCO}}
  \providecommand{\CNFISAAC}{\CNFX{ISAAC}}
  \providecommand{\CNFLICS}{\CNFX{LICS}}
  \providecommand{\CNFPODS}{\CNFX{PODS}}
  \providecommand{\CNFSWAT}{\CNFX{SWAT}}
  \providecommand{\CNFWADS}{\CNFX{WADS}}

% flatex input end: [./ann_revisit.bbl]
}%
 
%\bibliographystyle{alpha}%
%FLATEX-REM:\bibliography{ann_revisit}
%\bibliography{shortcuts,geometry}

\InsertAppendixOfProofs

\end{document}